\documentclass[11pt]{article}
\usepackage[margin=1in]{geometry}
\usepackage[utf8]{inputenc}
\usepackage{microtype}
\usepackage[all=normal,bibliography=tight]{savetrees}
\usepackage{listings}
  \usepackage{mathrsfs}

  \usepackage{amstext,amsfonts,amsthm,amsmath,amssymb}

\usepackage{graphicx,tikz}
\usepackage{comment}
\usepackage{url}
\usepackage{xspace}
\usepackage{todonotes}
\usepackage{enumerate}
\usetikzlibrary{shapes.geometric}

\usetikzlibrary{shapes.geometric}

\usepackage[absolute]{textpos}

\def\cqedsymbol{\ifmmode$\lrcorner$\else{\unskip\nobreak\hfil
\penalty50\hskip1em\null\nobreak\hfil$\lrcorner$
\parfillskip=0pt\finalhyphendemerits=0\endgraf}\fi}

\newtheorem{lemma}{Lemma}[section]

\newtheorem{theorem}[lemma]{Theorem}
\newtheorem{claim}[lemma]{Claim}
\theoremstyle{definition}
\newtheorem{definition}[lemma]{Definition}

\newcommand{\myparagraph}[1]{\paragraph{#1}}

\newcommand{\Oh}{\mathcal{O}}

\newcommand{\bitv}[1]{\ensuremath{[2]^{#1}}}
\newcommand{\mincut}[4]{\ensuremath{\operatorname{mincut}_{#1,#2}(#3,#4)}}
\newcommand{\R}{\mathbb{R}}
\newcommand{\rank}[1]{\ensuremath{\operatorname{rank}({#1})}}
\newcommand{\du}[1]{#1}
\newcommand{\duu}[1]{\ensuremath{{#1}^*}}
\newcommand{\term}{\ensuremath{Q}}

\newcommand{\tf}{\ensuremath{Q}}
\newcommand{\tfs}{\ensuremath{S}}
\newcommand{\sign}[1]{\ensuremath{\chi(#1)}}
\newcommand{\cycle}{\ensuremath{\mathcal{C}}}
\newcommand{\rev}[1]{\ensuremath{\overleftarrow{#1}}}
\newcommand{\dec}[1]{\ensuremath{\operatorname{dec}(#1)}}

\title{An exponential lower bound for cut sparsifiers in planar graphs\thanks{%
This research is a part of projects that have received funding from the European Research Council (ERC) under the European Union's Horizon 2020 research and innovation programme
under grant agreements No 714704 (Marcin Pilipczuk and Anna Zych-Pawlewicz).
Nikolai Karpov has been supported by the Warsaw Centre of Mathematics and Computer Science and the Government of the Russian Federation (grant 14.Z50.31.0030).}}

\author{ 
  Nikolai Karpov\thanks{St. Petersburg Department of V.A. Steklov Institute of Mathematics of the Russian Academy of Sciences, Russia and Institute of Informatics, University of Warsaw, Poland, \texttt{kimaska@gmail.com}.}
  \and 
  Marcin Pilipczuk\thanks{
    Institute of Informatics, University of Warsaw, Poland, \texttt{malcin@mimuw.edu.pl}.
  }
  \and 
  Anna Zych-Pawlewicz\thanks{
    Institute of Informatics, University of Warsaw, Poland, \texttt{anka@mimuw.edu.pl}.
  }
}

\date{}

\begin{document}

\maketitle

\begin{textblock}{20}(0, 12.5)
\includegraphics[width=40px]{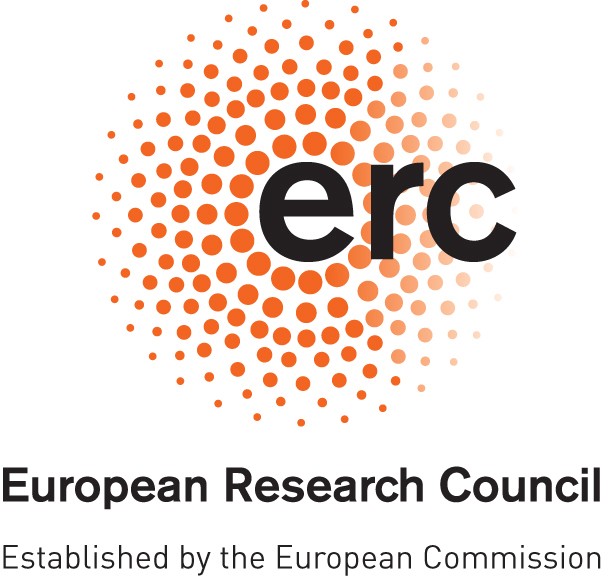}%
\end{textblock}
\begin{textblock}{20}(-0.25, 12.9)
\includegraphics[width=60px]{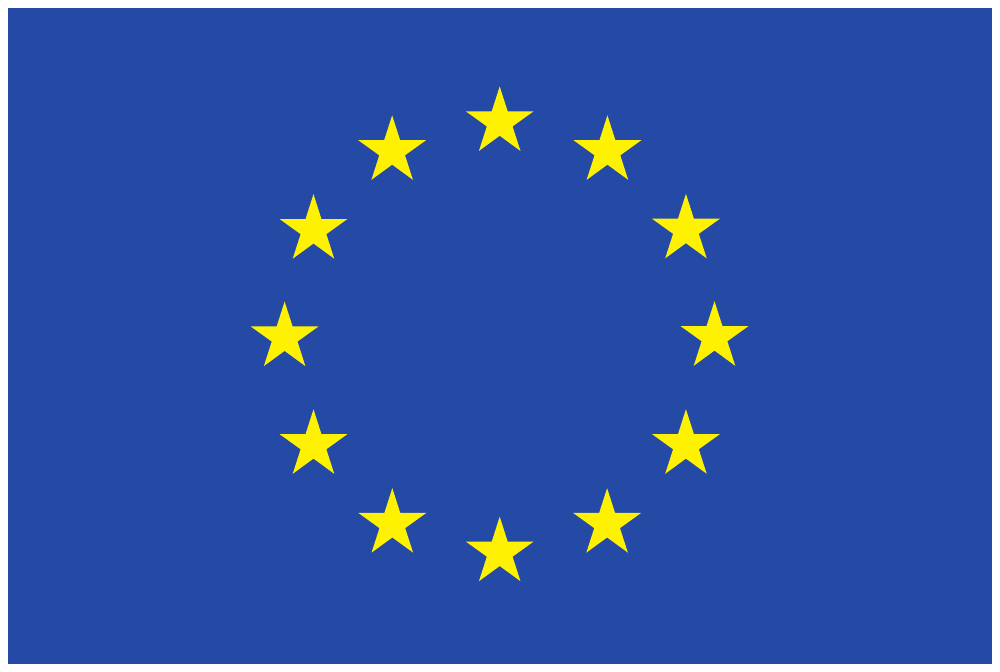}%
\end{textblock}

\begin{abstract}
Given an edge-weighted graph $G$ with a set $\term$ of $k$ terminals,
a \emph{mimicking network} is a graph with the same set of terminals that exactly preserves
the sizes of minimum cuts between any partition of the terminals.
A natural question in the area of graph compression is to provide as small mimicking networks as possible for
input graph $G$ being either an arbitrary graph or coming from a specific graph class.

In this note we show an exponential lower bound for cut mimicking networks in planar graphs:
there are edge-weighted planar graphs with $k$ terminals that require $2^{k-2}$ edges in any 
mimicking network.
This nearly matches an upper bound of $\Oh(k 2^{2k})$ of Krauthgamer and Rika [SODA 2013, arXiv:1702.05951]
and is in sharp contrast with the $\Oh(k^2)$ upper bound under the assumption that all terminals lie on a single
face [Goranci, Henzinger, Peng, arXiv:1702.01136].
As a side result we show a hard instance for double-exponential upper bounds given by Hagerup, Katajainen, Nishimura, and Ragde~[JCSS 1998],
   Khan and Raghavendra~[IPL 2014],
   and Chambers and Eppstein~[JGAA 2013].
\end{abstract}

\section{Introduction}\label{sec:intro}
One of the most popular paradigms when designing effective algorithms is preprocessing.
These days in many applications, in particular mobile ones, even though fast running time is desired, the memory usage is the main limitation. The preprocessing needed for such applications is to reduce the size of the input data prior to some resource-demanding computations,
 without (significantly) changing the answer to the problem being solved.
In this work we focus on this kind of preprocessing, known also as graph compression, for flows and cuts.
The input graph needs to be compressed while preserving its essential flow and cut properties.

Central to our work is the concept of a \emph{mimicking network}, introduced by 
Hagerup, Katajainen, Nishimura, and Ragde~\cite{HagerupKNR98}.
Let $G$ be an edge-weighted graph with a set $Q \subseteq V(G)$ of $k$ terminals.
For a partition $Q = S \uplus \bar{S}$, 
a minimum cut between $S$ and $\bar{S}$ is called a \emph{minimum $S$-separating cut}. 
A \emph{mimicking network} is an edge-weighted graph
$G'$ with $Q \subseteq V(G')$ such that the weights of minimum $S$-separating cuts
are equal in $G$ and $G'$ for every partition $Q = S \uplus \bar{S}$.
Hagerup et al~\cite{HagerupKNR98} 
observed the following simple preprocessing step: if two vertices $u$ and $v$
are always on the same side of the minimum cut between $S$ and $\bar{S}$ for every choice
of the partition $Q = S \uplus \bar{S}$, then they can be merged without changing the size
of any minimum $S$-separating cut. 
This procedure always
leads to a mimicking network with at most $2^{2^k}$ vertices. 

The above upper bound can be improved to a still double-exponential bound
of roughly $2^{\binom{k-1}{\lfloor (k-1)/2 \rfloor}}$, as observed both by 
Khan and Raghavendra~\cite{KhanR14} and by Chambers and Eppstein~\cite{ChambersE13}.
In 2013, Krauthgamer and Rika~\cite{KrauthgamerR13} observed that the aforementioned preprocessing
step can be adjusted to yield a mimicking network of size $\Oh(k^2 2^{2k})$ for planar graphs.
Furthermore, they introduced a framework for proving lower bounds, and showed that
there are (non-planar) graphs, for which any mimicking network 
has $2^{\Omega(k)}$ edges; a slightly stronger lower bound 
of $2^{(k-1)/2}$ has been shown by Khan and Raghavendra~\cite{KhanR14}.
On the other hand, for planar graphs the lower bound of~\cite{KrauthgamerR13} is $\Omega(k^2)$. 
Furthermore, the planar graph lower bound applies even in the special case when all the terminals
lie on the same face.

Very recently, two improvements upon these results for planar graphs have been announced. 
In a sequel paper, Krauthgamer and Rika~\cite{KrauthgamerR17} improve the 
polynomial factor in the upper bound for planar graphs to $\Oh(k 2^{2k})$ and show that the exponential dependency
actually adheres only to the \emph{number of faces containing terminals}: if
the terminals lie on $\gamma$ faces, one can obtain a mimicking network
of size $\Oh(\gamma 2^{2\gamma} k^4)$. 
In a different work, Goranci, Henzinger, and Peng~\cite{GoranciHP17} showed a tight $\Oh(k^2)$ upper bound
for mimicking networks for planar graphs with all terminals on a single face.

\myparagraph{Our results.}
We complement these results by showing an exponential lower bound for mimicking networks in planar graphs.
\begin{theorem}\label{thm:main}
For every integer $k \geq 3$,
there exists a planar graph $G$ with a set $Q$ of $k$ terminals
and edge cost function under which every mimicking network for $G$ has
at least $2^{k-2}$ edges.
\end{theorem}
This nearly matches the upper bound of $\Oh(k2^{2k})$ of  Krauthgamer and Rika~\cite{KrauthgamerR17}
and is in sharp contrast with the polynomial bounds when the terminals lie on a constant
number of faces~\cite{GoranciHP17,KrauthgamerR17}.
Note that it also nearly matches the improved bound of $\Oh(\gamma 2^{2\gamma} k^4)$ for terminals on $\gamma$ faces~\cite{KrauthgamerR17},
as $k$ terminals lie on at most $k$ faces.

As a side result, we also show a hard instance for mimicking networks in general graphs.
\begin{theorem}\label{thm:side}
For every integer $k \geq 1$ that is equal to $6$ modulo $8$, there
exists a graph $G$ with a set $Q$ of $k$ terminals 
and $\Omega(2^{\binom{k-1}{\lfloor (k-1)/2 \rfloor} - k/2})$ vertices,
    such that no two vertices can be identified without strictly increasing the size of some minimum $S$-separating cut.
\end{theorem}
The example of Theorem~\ref{thm:side}, obtained by iterating the construction of Krauthgamer and Rika~\cite{KrauthgamerR13},
shows that the doubly exponential bound
is natural for the preprocessing step of Hagerup et al~\cite{HagerupKNR98}, and
one needs different techniques to improve upon it.
Note that the bound of Theorem~\ref{thm:side} is very close to the upper bound given
by~\cite{ChambersE13,KhanR14}.

\myparagraph{Related work.}
Apart from the aforementioned work on mimicking
networks~\cite{GoranciHP17,HagerupKNR98,KhanR14,KrauthgamerR13,KrauthgamerR17},
there has been substantial work on preserving cuts and flows approximately,
see e.g.~\cite{robi-apx,robi-old,mm-sparsifiers}.
If one wants to construct mimicking networks for vertex cuts in
unweighted graphs with deletable terminals (or with small integral
weights), the representative sets approach of Kratsch and Wahlstr\"{o}m~\cite{KratschW12}
provides a mimicking network with $\Oh(k^3)$ vertices, improving upon a previous
quasipolynomial bound of Chuzhoy~\cite{Chuzhoy12}.


\medskip

We prove Theorem~\ref{thm:main} in Section~\ref{sec:main}
and show the example of Theorem~\ref{thm:side} in Section~\ref{sec:side}.
%

\section{Exponential lower bound for planar graphs}\label{sec:main}
In this section we present the main result of the paper.
We provide a construction that proves that there are planar graphs with $k$ terminals whose mimicking networks are of size $\Omega(2^k)$. 

In order to present the desired graph, for the sake of simplicity, we describe its dual graph $(\du{G},\du{c})$. We let $\du{\tf}=\{\du{f_n},\du{f_s},\du{f_1},\du{f_2},\dots,\du{f_{k-2}} \}$ be the set of faces in $\du{G}$ corresponding to terminals in the primal graph $\duu{G}$.%
\footnote{Since the argument mostly operates on the dual graph, for notational simplicity,
  we use regular symbols  for objects in the dual graph, e.g., $G$, $c$, $f_i$,
  while starred symbols refer to the dual of the dual graph, that is, the primal graph.}
There are two special terminal faces $\du{f_n}$ and $\du{f_s}$, referred to as the north face and the south face. The remaining faces of $\du{\tf}$ are referred to as equator faces.

A set $\du{\tfs} \subset \du{\tf}$ is \emph{important} if $\du{f_n} \in \du{\tfs}$ and $\du{f_s} \notin \du{\tfs}$. Note that there are $2^{k-2}$ important sets; in what follows we care only
about minimum cuts in the primal graph for separations between important sets and their complements.
For an important set $\du{\tfs}$, 
we define its \emph{signature} as a bit vector $\sign{\du{\tfs}} \in \bitv{|\du{\tf}|-2}$ whose $i$'th position is defined as $\sign{\du{\tfs}}[i]= 1 \text{ iff } \du{f_{i}} \in \du{\tfs}$. 
Graph $\du{G}$ will be composed of $2^{k-2}$ cycles referred to as important cycles, each corresponding to an important subset $\du{\tfs} \subset \du{\tf}$.
A cycle corresponding to $\du{\tfs}$ is referred to as $\cycle_{\sign{\du{\tfs}}}$ and it separates $\du{\tfs}$ from $\overline{\du{\tfs}}$.
Topologically, we draw the equator faces on a straight horizontal line that we call the equator. We put the north face $\du{f_n}$ above the equator and the south face $\du{f_s}$ below the equator. For any important $\du{\tfs} \subset \du{\tf}$, in the plane drawing of $\du{G}$ the corresponding cycle $\cycle_{\sign{\du{\tfs}}}$ is a curve that goes to the south of $\du{f_i}$ if $\du{f_i} \in \du{\tfs}$ and otherwise to the north of $\du{f_i}$. We formally define important cycles later on, see Definition~\ref{def:impcyc}.

We now describe in detail the construction of $\du{G}$.
We start with a graph $H$ that is almost a tree, and then embed $H$ in the plane
with a number of edge crossings, introducing a new vertex on every edge crossing.
The graph $H$ consists of a complete binary tree of height $k-2$ with root $v$ and an extra vertex
$w$ that is adjacent to the root $v$ and every one of the $2^{k-2}$ leaves of the tree.
In what follows, the vertices of $H$ are called \emph{branching vertices}, contrary
to \emph{crossing vertices} that will be introduced at edge crossings
in the plane embedding of $H$.

To describe the plane embedding of $H$, we need to introduce some notation of the vertices
of $H$.
The starting point of our construction is the edge $\du{e}=\{ \du{w}, \du{v} \}$.
Vertex $\du{v}$ is the first branching vertex and also the root of $H$.
In vertex $\du{v}$, edge $\du{e}$ branches into $\du{e_0}=\{\du{v},\du{v_0}\}$ and $\du{e_1}=\{\du{v},\du{v_1} \}$. Now $\du{v_0}$ and $\du{v_1}$ are also branching vertices.
The branching vertices are partitioned into layers $L_0,\ldots,L_{k-2}$. Vertex $\du{v}$ is in layer $L_0=\{ \du{v} \}$, while $\du{v_0}$ and $\du{v_1}$ are in layer $L_1=\{ \du{v_0}, \du{v_1} \}$.  Similarly, we partition edges into layers $\mathcal{E}^H_0,\ldots \mathcal{E}^H_{k-1}$. So far we have $\mathcal{E}^H_0=\{ \du{e} \}$ and $\mathcal{E}^H_1=\{ \du{e_0}, \du{e_1} \}$. 

The construction continues as follows. For any layer $L_i, i \in \{1, \ldots , k-3 \}$, all the branching vertices of $L_i=\{ \du{v_{00 \ldots 0}} \ldots \du{v_{11 \ldots 1}} \}$ are of degree $3$. In a vertex $\du{v_a} \in L_i$, $a \in \bitv{i}$, edge $\du{e_a} \in \mathcal{E}^H_i$ branches into edges $\du{e_{0a}}=\{ \du{v_a}, \du{v_{0a}} \},\du{e_{1a}}=\{ \du{v_a}, \du{v_{1a}} \} \in \mathcal{E}^H_{i+1}$, where $\du{v_{0a}},\du{v_{1a}} \in L_{i+1}$. We emphasize here that the new bit in the index is added \emph{as the first symbol}. 
Every next layer is twice the size of the previous one, hence $|L_i|=|\mathcal{E}^H_i|=2^i$. Finally the vertices of $L_{k-2}$ are all of degree $2$. Each of them is connected to a vertex in $L_{k-3}$ via an edge in $\mathcal{E}^H_{k-2}$ and to the vertex $w$ via an edge in $\mathcal{E}^H_{k-1}$.

We now describe the drawing of $H$, that we later make planar by adding crossing vertices, in order to obtain the graph $G$.
As we mentioned before, we want to draw equator faces $\du{f_1}, \ldots \du{f_{k-2}}$ in that order from left to right on a horizontal line (referred to as an equator). Consider equator face $\du{f_i}$ and vertex layer $L_i$ for some $i>0$. Imagine a vertical line through $\du{f_i}$ perpendicular to the equator, and let us refer to it as an $i$'th meridian. We align the vertices of $L_i$ along the $i$'th meridian, from the north to the south. We start with the vertex of $L_i$ with the (lexicographically) lowest index, and continue drawing vertices of $L_i$ more and more to the south while the indices increase. Moreover, the first half of $L_i$ is drawn to the north of $\du{f_i}$, and the second half to the south of $\du{f_i}$.
Every edge of $H$, except for $e$, is drawn as a straight line segment connecting its endpoints.
The edge $\du{e}$ is a curve encapsulating the north face $\du{f_n}$ and separating it from $\du{f_s}$-the outer face of $\du{G}$.

\begin{figure}[t]
\begin{center}
\begin{tikzpicture}[scale=0.9]

\node[circle, fill, label=above:$\du{v_0}$] at (0,1) {}; 
\node[circle, fill, label=above:$\du{v_1}$] at (0,-1) {};
\node[circle, fill, label=above:$\du{v_{01}}$] at (2.2,1) {}; 
\node[circle, fill, label=above:$\du{v_{00}}$] at (2.2,2) {}; 
\node[circle, fill, label=above:$\du{v_{10}}$] at (2.2,-1) {}; 
\node[circle, fill, label=above:$\du{v_{11}}$] at (2.2,-2) {};
\node[circle, fill, label=above:$\du{v_1}$] at (0,-1) {};
\foreach \i [evaluate=\i as \p using ((2)^(\i)), evaluate=\i as \q using int(1+\i)] in {0,...,2}{
	\foreach \j [evaluate=\j as \t using \j + \p] in {-\p,...,-1}{
		\node[circle,fill] at (2.2*\i, \j) {};	
	}
	\foreach \j [evaluate=\j as \t using \j + \p-1] in {1,...,\p}{
		\node[circle,fill] at (2.2*\i, \j) {};	
	}
	\node[diamond] at (2.2*\i, 0) {$\du{f_{\q}}$};
}
\node[circle, fill, label=left:{$\du{v}$}] (-1;0) at (-2.2, 0) {};

\draw (2.2*0, 1) -- (2.2*1,2) node [midway, above, sloped] {$\du{e_{00}}$};
\draw (2.2*1, 2) -- (2.2*2,4) node [midway, above, sloped] {$\du{e_{000}}$};

\foreach \i [evaluate=\i as \p using (2)^(\i)] in {0, ...,1}{
	\foreach \j in {-\p, ..., -1}{
		\draw (2.2*\i,\j) -- ({2.2*(\i+1)},\j-\p);
		\draw (2.2*\i,\j) -- ({2.2*(\i+1)},\j+\p+1);
	}
	\draw (2.2*0,-1) -- (2.2,-2) node [midway, below,sloped] {$\du{e_{11}}$};
	\draw (2.2*1,-2) -- (2.2*2,-4) node [midway, below,sloped] {$\du{e_{111}}$};
	\foreach \j in {1, ..., \p}{
		\draw (2.2*\i, \j) -- ({2.2*(\i+1)}, \j-\p-1);
		\draw (2.2*\i, \j) -- ({2.2*(\i+1)}, \j+\p);
	}
	\draw (-2.2, 0) -- (0, 1) node [midway,above, sloped] {$\du{e_0}$};
	\draw (-2.2, 0) -- (0, -1) node [midway,below=1pt] {$\du{e_1}$};
	\foreach \j [evaluate=\j as \t using \j + 8]in {-8,...,-1}{
		\node[circle,fill] at (2.2*3.5, \j*.65) {};	
	}
	\foreach \j [evaluate=\j as \t using \j + 8-1] in {1,...,8}{
		\node[circle,fill] at (2.2*3.5, \j*.65) {};	
	}
	\node[diamond] at (2.2*3.5+0.1, 0) {$\du{f_{k-2}}$};
	\node[circle, fill, label=right:{$\du{w}$}] at (2.2*5.5, 0) {};
	
	\foreach \j [evaluate=\j as \t using \j + 8]in {-8,...,-1}{
	}
	\foreach \j [evaluate=\j as \t using \j + 8-1] in {1,...,8}{
	}
	\foreach \j in {-4, ..., -1}{
		\draw plot [smooth] coordinates {(2.2*2, \j) ({2.2*3.5}, {(\j-4)*.65}) (2.2 * 5.5, 0)};
		\draw plot [smooth] coordinates {(2.2*2, \j) ({2.2*3.5}, {(\j+5)*.65}) (2.2 * 5.5, 0)};
	}
	\foreach \j in {1, ..., 4}{
		\draw plot [smooth] coordinates {(2.2*2, \j) ({2.2*3.5}, {(\j+4)*.65}) (2.2 * 5.5, 0)};
		\draw plot [smooth] coordinates {(2.2*2, \j) ({2.2*3.5}, {(\j-4-1)*.65}) (2.2 * 5.5, 0)};
	}
	\draw (5+1.1*6.5, 0) arc (0:180:1.1*6.5) node [midway, above=2pt] {$\du{e}$};
	\node at (2.2*6.5/2-4, 2.2*6.5/2-2) {$\du{f_n}$};
	\node at (2.2*6.5/2-8, 2.2*6.5/2-1) {$\du{f_s}$};
	\draw[dashed] (-2.2, 0) -- (-2.2, -6) node [right] {$L_0$};
	\foreach \i[evaluate=\i as \ii using int(\i+1)]in {0,...,2}{
		\draw[dashed] (\i*2.2,-2^\i) -- (\i * 2.2, -6) node [right] {$L_{\ii}$};
	}
	\foreach \i in {0, ..., 3}{
		\node at (2.2*\i-3.2, -5) {$\mathcal{E}_{\i}$};
	}
}

\end{tikzpicture}
\caption{The graph $\du{G}$.\label{dual}}
\end{center}
\end{figure}

The crossing vertices are added whenever the line segments cross. This way the edges of $H$
are subdivided and the resulting graph is denoted by $\du{G}$.
This completes the description of the structure and the planar drawing of $\du{G}$.
We refer to Figure~\ref{dual} for an illustration of the graph $G$.
The set $\mathcal{E}_i$ consists of all edges of $G$ that are parts of the (subdivided) edges of $\mathcal{E}^H_i$ from $H$, see Figure~\ref{subdivide}.
We are also ready to define important cycles formally.

\begin{figure}[t]
\centering
\begin{tikzpicture}[scale=0.7]

{\tiny
\draw[draw=white] (4 + 0, 1) -- (4 + 6, 5*0.8) node [pos=0.07, above,sloped] {$\du{e^{1}_{0a}}$};
\draw[draw=white] (4 + 0, 1) -- (4 + 6, 5*0.8) node [pos=0.2, above,sloped] {$\du{e^{2}_{0a}}$};
\draw[draw=white] (4 + 0, 1) -- (4 + 6, 5*0.8) node [pos=0.3, above,sloped] {$\cdots$};
\draw[draw=white] (4 + 0, 1) -- (4 + 6, 5*0.8) node [pos=0.5, above,sloped] {$\du{e^{\mathrm{dec}(a)}_{0a}}$};

\draw[draw=white] (4 + 0, 4) -- (4 + 6, -1*0.8) node [pos=0.1, above, sloped, text=blue] {$c_{i+1}$};
\draw[draw=white] (4 + 0, 4) -- (4 + 6, -1*0.8) node [pos=0.23, above, sloped, text=blue] {$c_{i+1}$};

\draw[draw=white] (4 + 0, 1) -- (4 + 6, -4*0.8) node [pos=0.1, above, sloped] {$\du{e^{1}_{1a}}$};
\draw[draw=white] (4 + 0, 1) -- (4 + 6, -4*0.8) node [pos=0.3, above, sloped] {$\du{e^{2}_{1a}}$};
\draw[draw=white] (4 + 0, 1) -- (4 + 6, -4*0.8) node [pos=0.45, above, sloped] {$\cdots$};
\draw[draw=white] (4 + 0, 1) -- (4 + 6, -4*0.8) node [pos=0.7, above, sloped] {$\du{e^{2^i-\mathrm{dec}(a)}_{1a}}$};
\foreach \i in {1,...,4}{
	\node[circle,fill] at (4 + 0, \i) {};
}
\foreach \i in {1,...,4}{
	\node[circle,fill] at (4 + 0, -\i) {};
}
\foreach \i in {1,...,8}{
	\node[circle,fill] at (4 + 6, \i*0.8) {};
}
\foreach \i in {1,...,8}{
	\node[circle,fill] at (4 + 6, -\i*0.8) {};
}
\foreach \i in {1, ..., 4} {
	\draw (4 + 0, \i) -- (4 + 6,{(\i+4)*0.8}) node [pos=0.9, above, text=blue] {$C$};
}
\foreach \i in {1, ..., 4} {
	\draw (4 + 0, -\i) -- (4 + 6,{(5-\i)*0.8}) node [pos=0.9,above, text=blue] {$C$};
}

\foreach \i in {1, ..., 4} {
	\draw (4 + 0, -\i) -- (4 + 6,{-(\i+4)*0.8}) node [pos=0.9, below, text=blue] {$C$};
}
\foreach \i in {1, ..., 4} {
	\draw (4 + 0, \i) -- (4 + 6,{-(5-\i)*0.8}) node [pos=0.9,below, text=blue] {$C$};
}

\draw[dashed] (4 + 0, -4) -- (4 + 0, -8) node[right] {$L_i$};
\draw[dashed] (4 + 6, -8*0.8) -- (4 + 6, -8) node[right] {$L_{i+1}$};
\node[circle] at (4 + 3, -7) {$\mathcal{E}_{i+1}$};

\node[label=right:{$\du{v_{0a}}$}] at (4 + 6, 5*0.8) {};
\node[label=left:{$\du{v_{a}}$}] at (4 + 0, 1) {};
\node[label=right:{$\du{v_{1a}}$}] at (4 + 6, -3) {};

}

\end{tikzpicture}
\caption{The layer $\mathcal{E}_{i+1}$.
  The vertex and edge names are black, their weights are blue.\label{subdivide}}
\end{figure}
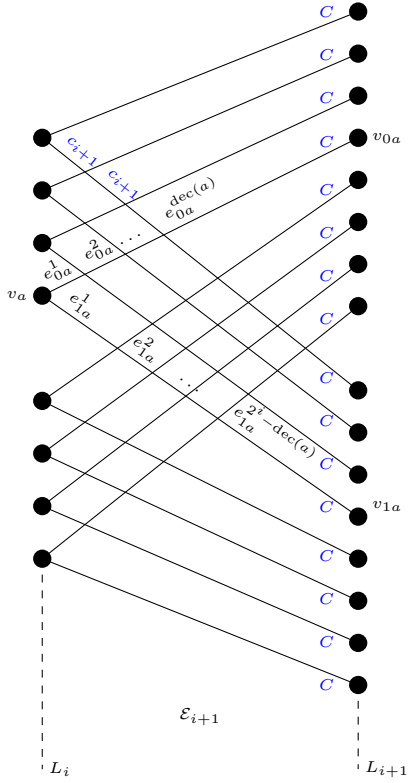

\begin{definition}\label{def:impcyc}
Let $\du{\tfs} \subset \du{\tf}$ be important.
Let $\pi$ be a unique path in the binary tree $H-\{w\}$ from the root
$\du{v}$ to $\du{v_{\rev{\sign{\du{\tfs}}}}}$, 
where $\rev{\cdot}$ operator reverses the bit vector.
Let $\pi'$ be the path in $G$ corresponding to $\pi$.
The important cycle $\cycle_{\sign{\du{\tfs}}}$ is composed of $\du{e}$, $\pi'$, and an edge in $\mathcal{E}_{k-1}$ adjacent to $\du{v_{\rev{\sign{\du{\tfs}}}}}$.  
\end{definition}

We now move on to describing how weights are assigned to the edges of $\du{G}$. 
The costs of the edges in $\du{G}$ admit $k-1$ values: $c_1, c_2, \ldots c_{k-2}$, and $C$. Let $c_{k-2}=1$. For $i \in \{1 \dots k-3 \}$ let $c_i= \sum_{j=i+1}^{k-2}|\mathcal{E}_{j}|c_{j}$.
Let $C=\sum_{j=1}^{k-2} |\mathcal{E}_i|c_i$. Let us consider an arbitrary edge $\du{e_{ba}}=\{ \du{v_{a}}, \du{v_{ba}} \}$ for some $a \in \bitv{i}, i \in \{ 0 \ldots k-3 \}, b \in \{ 0,1 \}$ (see Figure~\ref{subdivide} for an illustration). As we mentioned before, $\du{e_{ba}}$ is subdivided by crossing vertices into a number of edges. If $b=0$, then edge $\du{e_{ba}}$ is subdivided by\footnote{For a bit vector $a$, $\dec{a}$ denotes the integral value of $a$ read as a number in binary.} $\dec{a}$ crossing vertices into $\dec{a}+1$ edges: $\du{e^1_{ba}}=\{ \du{v_a}, \du{x^1_{ba}} \}, \du{e^2_{ba}}=\{ \du{x^1_{ba}},\du{x^2_{ba}} \} \ldots \du{e^{\dec{a}+1}_{ba}}=\{ \du{x^{\dec{a}}_{ba}}, \du{v_{ba}} \}$. Among those edges $\du{e^{\dec{a}+1}_{ba}}$ is assigned cost $C$, and the remaining edges subdividing $\du{e_{ba}}$ are assigned cost $c_i$. Analogically, if $b=1$, then edge $\du{e_{ba}}$ is subdivided by $2^i-1-\dec{a}$ crossing vertices into $2^i-\dec{a}$ edges: $\du{e^1_{ba}}=\{ \du{v_a}, \du{x^1_{ba}} \}, \du{e^2_{ba}}=\{ \du{x^1_{ba}},\du{x^2_{ba}} \} \ldots \du{e^{2^i-\dec{a}}_{ba}}=\{ \du{x^{2^i-1-\dec{a}}_{ba}}, \du{v_{ba}} \}$. Again, we let edge $\du{e^{2^i-\dec{a}}_{ba}}$ have cost $C$, and the remaining edges subdividing $\du{e_{ba}}$ are assigned cost $c_i$.
Finally, all the edges connecting the vertices of the last layer with $w$ have weight $c_{k-2} = 1$.
The cost assignment within an edge layer is presented in Figure~\ref{subdivide}. 

This finishes the description of the dual graph $G$. We now consider the primal graph $\duu{G}$ with the set of terminals $\duu{\tf}$ consisting of the
$k$ vertices of $\duu{G}$ corresponding to the faces $\tf$ of $G$. In the remainder of this section we show that there is a cost function on the edges of $\duu{G}$, under which any mimicking network for $\duu{G}$
contains at least $2^{k-2}$ edges. This cost function is in fact a small perturbation of the edge costs implied by the dual graph $G$.

In order to accomplish this,
we use the framework introduced in~\cite{KrauthgamerR13}. In what follows,
   $\mincut{G}{c}{S}{S'}$ stands for the minimum cut separating $S$ from $S'$ in a graph $G$ with cost function $c$. Below we provide the definition of the cutset-edge incidence matrix and the Main Technical Lemma from~\cite{KrauthgamerR13}. 

\begin{definition}[Incidence matrix between cutsets and edges] Let $(G,c)$ be a $k$-terminal network, and fix an enumeration $S_1, \ldots S_m$ of all $2^{k-1}-1$ distinct and nontrivial bipartitions $Q=S_i \cup \overline{S}_i$. The cutset-edge incidence matrix of $(G,c)$ is the matrix $A_{G,c} \in \{ 0,1 \}^{m \times E(G)}$ given by
$$
(A_{G,c})_{i,e}=
\begin{cases}
1 \text{ if } e \in \mincut{G}{c}{S_i}{\overline{S}_i}\\
0 \text{ otherwise.}
\end{cases}
$$
\end{definition}

\begin{lemma}[Main Technical Lemma of \cite{KrauthgamerR13}]\label{lem:mtl}
Let $(G,c)$ be a $k$-terminal network. Let $A_{G,c}$ be its cutset-edge incidence matrix, and assume that for all $S \subset Q$ the minimum $S$-separating cut of $G$ is unique. Then there is for $G$ an edge cost function $\tilde{c}: E(G) \mapsto \R^+$, under which every mimicking network $(G',c')$ satisfies $|E(G')| \geq \rank{A_{G,c}}$. 
\end{lemma}

Recall that $\duu{G}$ is the dual graph to the graph $\du{G}$ that we constructed.
By slightly abusing the notation, we will use the cost function $c$ defined on the dual edges
also on the corresponding primal edges.
Let $\duu{\tf}=\{ \duu{f_n}, \duu{f_s}, \duu{f_1}, \ldots \duu{f_{k-2}} \}$ be the set of terminals in $\duu{G}$ corresponding to $\du{f_n}, \du{f_s}, \du{f_1}, \ldots \du{f_{k-2}}$ respectively. We want to apply Lemma~\ref{lem:mtl} to $\duu{G}$ and $\duu{\tf}$. For that we need to show that the cuts in $\duu{G}$ corresponding to important sets are unique and that $\rank{A_{\duu{G},c}}$ is high.

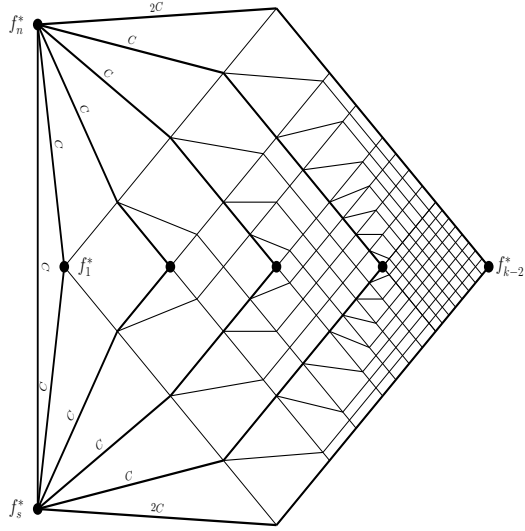
\begin{figure}[tb]
\begin{center}
\resizebox{200pt}{200pt}{\begin{tikzpicture}
\foreach \j in {0, ..., 4}{
	\node[circle, fill] at (\j*4, 0) {};
}
\node [circle, fill] at (-1, 2.5*3) {};
\node [circle, fill] at (-1, -2.5*3) {};
{\LARGE \draw (-1.8, 2.5*3) node {$f_n^\ast$};
\draw (-1.8, -2.5*3) node {$f_s^\ast$};
\draw (0.8, 0) node {$f_1^\ast$};
\draw (16.8, 0) node {$f_{k-2}^\ast$};
}
\draw[line width=2pt] (-1, 2.5*3) -- (-1, -2.5*3)node [midway, above, sloped] {$C$};

\foreach \i in {0, ..., 4}{
	\draw (2*\i,2*\i) -- (8 + 2*\i, -8 + 2*\i);
	\draw (2*\i,-2*\i) -- (8+ 2*\i,8 + -2*\i);
	\draw[line width=2pt] (2*\i, -2*\i) -- ({4*(\i)}, 0);
	\draw[line width=2pt] (2*\i, 2*\i) -- ({4*(\i)}, 0);
}

\foreach \i in {0, ..., 3}{
	\draw[line width=2pt] (-1, 2.5*3) -- (2*\i, 2*\i) node [midway, above, sloped] {$C$};
	\draw[line width=2pt] (-1, -2.5*3) -- (2*\i, -2*\i) node [midway, above, sloped] {$C$};
}
\foreach \i in {4}{
	\draw[line width=2pt] (-1, 2.5*3) -- (2*\i, 2*\i) node [midway, above, sloped] {$2C$};
	\draw[line width=2pt] (-1, -2.5*3) -- (2*\i, -2*\i) node [midway, above, sloped] {$2C$};
}

\foreach \i [evaluate=\i as \p using 2^(\i)] in {1, ..., 3}{
	\draw plot coordinates {(2*\i, 2*\i) (2*\i+4 - 2/\p, 2*\i - 2/\p) (4*\i+2 - 2/\p, 2 - 2/\p)};
	\draw plot coordinates {(2*\i, -2*\i) (2*\i+4 -2/\p, -2*\i + 2/\p) (4*\i+2 - 2/\p, -2 + 2/\p)};
  \draw[] (4*\i+2 - 2/\p, 2 - 2/\p) -- (4*\i+4 - 2/\p, 0 - 2/\p);
  \draw[] (4*\i+2 - 2/\p, -2 + 2/\p) -- (4*\i+4 - 2/\p, 0 + 2/\p);
}
\foreach \i [evaluate=\i as \p using 2^(\i)] in {2, ..., 3}{
	\draw plot coordinates {(2*\i + 2, 2*\i - 2) (2*\i+5-4/\p, 2*\i-1-4/\p) (4*\i + 2 -4/\p, 2 -4/\p)};
	\draw (2*\i+3, 2*\i-3) -- (2*\i+5, 2*\i-1);
	\draw plot coordinates {(2*\i + 3, 2*\i - 3) (2*\i + 6 - 6/\p, 2*\i - 2 - 6/\p) (4*\i + 2 - 6/\p, 2 - 6/\p)};

	\draw plot coordinates {(2*\i + 2, -2*\i + 2) (2*\i+5-4/\p, -2*\i+1+4/\p) (4*\i + 2 -4/\p, -2 +4/\p)};
	\draw (2*\i+3, -2*\i+3) -- (2*\i+5, -2*\i+1);
	\draw plot coordinates {(2*\i + 3, -2*\i + 3) (2*\i + 6 - 6/\p, -2*\i + 2 + 6/\p) (4*\i + 2 - 6/\p, -2 + 6/\p)};

  \draw[] (4*\i + 2 -4/\p, 2 -4/\p) -- (4*\i + 4-4/\p, 0-4/\p);
  \draw[] (4*\i + 2 - 6/\p, 2 - 6/\p) -- (4*\i + 4 - 6/\p, 0 - 6/\p);
  \draw[] (4*\i + 2 -4/\p, -2 +4/\p) -- (4*\i + 4 -4/\p, 0 +4/\p);
  \draw[] (4*\i + 2 - 6/\p, -2 + 6/\p) -- (4*\i + 4 - 6/\p, 0 + 6/\p);
}
\foreach \i [evaluate=\i as \p using 2^(\i)] in {3, ..., 3}{
	\foreach \j in {0, ..., 3}{
		\draw (2*\i+4 + \j/2, 2*\i-4 -\j/2)--
		(2*\i + 2 +4 + \j/2 + 1/2 - 8/\p- 2 * \j / \p, 2*\i + 2 -4 -\j/2 - 1/2 - 8/\p - 2 * \j / \p) -- 
		(2*\i + 4 -0.5 +4 + 1/2 - 8/\p- 2 * \j / \p, 2*\i + 0.5 -4 - 1/2 - 8/\p - 2 * \j / \p);
		\draw (2*\i+4 + \j/2 + 1/2, 2*\i-4 -\j/2 - 1/2) -- (2*\i+4 + \j/2 + 1/2 + 2, 2*\i-4 -\j/2 - 1/2 + 2);
    
    \draw[] (2*\i + 4 -0.5 +4 + 1/2 - 8/\p- 2 * \j / \p, 2*\i + 0.5 -4 - 1/2 - 8/\p - 2 * \j / \p) --
    (2*\i + 4 -0.5 +4 + 1/2+2 - 8/\p- 2 * \j / \p, 2*\i + 0.5 -4-2 - 1/2 - 8/\p - 2 * \j / \p);

		\draw (2*\i+4 + \j/2, -2*\i+4 +\j/2)--
		(2*\i + 2 +4 + \j/2 + 1/2 - 8/\p- 2 * \j / \p, -2*\i - 2 +4 +\j/2 + 1/2 + 8/\p + 2 * \j / \p) -- 
		(2*\i + 4 -0.5 +4 + 1/2 - 8/\p- 2 * \j / \p, -2*\i - 0.5 +4 + 1/2 + 8/\p + 2 * \j / \p);
		\draw (2*\i+4 + \j/2 + 1/2, -2*\i+4 +\j/2 + 1/2) -- (2*\i+4 + \j/2 + 1/2 + 2, -2*\i+4 +\j/2 + 1/2 - 2);

    \draw[] (2*\i + 4 -0.5 +4 + 1/2 - 8/\p- 2 * \j / \p, -2*\i - 0.5 +4 + 1/2 + 8/\p + 2 * \j / \p) --
     (2*\i + 4 -0.5 +4 + 1/2+2 - 8/\p- 2 * \j / \p, -2*\i - 0.5 +4 + 1/2+2 + 8/\p + 2 * \j / \p);
	}
}

\foreach \i [evaluate=\i as \p using 2^(\i)] in {1, 2, 3} {

}

\end{tikzpicture}
}
\caption{Primal graph $G^\ast$.\label{primal}}
\end{center}
\end{figure}

As an intermediate step let us argue that the following holds.
\begin{claim}\label{clm:kc}
There are $k$ edge disjoint simple paths in $\duu{G}$ from $\duu{f_n}$ to $\duu{f_s}$: $\pi_0, \pi_1, \ldots, \pi_{k-2}, \pi_{k-1}$. Each $\pi_i$ is composed entirely of edges dual to the edges of $\mathcal{E}_i$ whose cost equals $C$. For $i \in \{ 1 \ldots k-2 \}$, $\pi_i$ contains vertex $\duu{f_i}$. Let $\pi_i^n$ be the prefix of $\pi_i$ from $\duu{f_n}$ to $\duu{f_i}$ and $\pi_i^s$ be the suffix from $\duu{f_i}$ to $\duu{f_s}$. The number of edges on $\pi_i$ is $2^i$, and the number of edges on $\pi_i^n$ and $\pi_i^s$ is $2^{i-1}$. 
\end{claim}
\begin{proof}The primal graph $\duu{G}$ together with paths $\pi_0, \pi_1 \ldots \pi_{k-2},\pi_{k-1}$ is pictured in Figure~\ref{primal}. The paths $\pi_{k-2},\pi_{k-1}$ visit the same vertices in the same manner, so for the sake of clarity only one of these paths is shown in the picture. This proof contains a detailed description of these paths and how they emerge from in the dual graph $\du{G}$.

Consider a layer $L_i$. Recall that for any $ba \in \bitv{i}$ edge $\du{e_{ba}}$ of the almost tree is subdivided in $\du{G}$, and all the resulting edges are in $\mathcal{E}_i$. If $b=0$, then edge $\du{e_{ba}}$ is subdivided by $\dec{a}$ crossing vertices into $\dec{a}+1$ edges: $\du{e^1_{ba}}=\{ \du{v_a}, \du{x^1_{ba}} \}, \du{e^2_{ba}}=\{ \du{x^1_{ba}},\du{x^2_{ba}} \} \ldots \du{e^{\dec{a}+1}_{ba}}=\{ \du{x^{\dec{a}}_{ba}}, \du{v_{ba}} \}$, where $\du{c}(\du{e^{\dec{a}+1}_{ba}})=C$. Analogically, if $b=1$, then edge $\du{e_{ba}}$ is subdivided by $2^i-1-\dec{a}$ crossing vertices into $2^i-\dec{a}$ edges: $\du{e^1_{ba}}=\{ \du{v_a}, \du{x^1_{ba}} \}, \du{e^2_{ba}}=\{ \du{x^1_{ba}},\du{x^2_{ba}} \} \ldots \du{e^{2^i-\dec{a}}_{ba}}=\{ \du{x^{2^i-1-\dec{a}}_{ba}}, \du{v_{ba}} \}$. Again, $\du{c}(\du{e^{2^i-\dec{a}}_{ba}})=C$. Consider the edges of $\mathcal{E}_i$ incident to vertices in $L_i$. If we order these edges lexicographically by their lower index, then each consecutive pair of edges shares a common face. Moreover, the first edge $\du{e^1_{00\ldots0}}$ is incident to $\du{f_n}$ and the last edge $\du{e^1_{11\ldots1}}$ is incident to $\du{f_s}$. This gives a path $\pi_i$ from $f_n$ to $f_s$ through $f_i$ in the primal graph where all the edges on $\pi_i$ have cost $C$. Path $\pi_{k-1}$ is given by the edges of $\mathcal{E}_{k-1}$ in a similar fashion and path $\pi_0$ is composed of a single edge dual to $\du{e}$. 
\end{proof}

We move on to proving that the condition in Lemma~\ref{lem:mtl} holds.
We extend the notion of important sets $\tfs \subseteq \tf$ to sets $\duu{\tfs} \subseteq \duu{\tf}$
in the natural manner.
\begin{lemma}\label{lem:uniquecuts}
For every important $\duu{\tfs} \subset \duu{\tf}$, the minimum cut separating $\duu{\tfs}$ from $\overline{\duu{\tfs}}$ is unique and corresponds to cycle $\cycle_{\sign{\du{\tfs}}}$ in $\du{G}$.   
\end{lemma}

\begin{proof}
Let $\cycle$ be the set of edges of $G$ corresponding to some
minimum cut between $\duu{\tfs}$ and $\duu{\overline{\tfs}}$ in $\duu{G}$.
Let $\tfs \subseteq \tf$ be the set of faces of $G$ corresponding to the set $\duu{\tfs}$.
We start by observing that the edges of $\duu{G}$ corresponding to $\cycle_{\sign{\du{\tfs}}}$
form a cut between $\duu{\tfs}$ and $\duu{\overline{\tfs}}$. Consequently,
     the total weight of edges of $\cycle$ is at most the total weight of the edges of
     $\cycle_{\sign{\du{\tfs}}}$.

By Claim~\ref{clm:kc}, $\cycle$ contains at least $k$ edges of cost $C$, at least one edge of cost $C$ per edge layer (it needs to hit an edge in every path $\pi_0 , \ldots \pi_{k-1}$). Note that $\cycle_{\sign{ \du{\tfs} }}$ contains exactly $k$ edges of cost $C$. We assign the weights in a way that $C$ is larger than all other edges in the graph taken together.
This implies that $\cycle$ contains exactly one edge of cost $C$ in every edge layer $\mathcal{E}_i$.
In particular, $\cycle$ contains the edge $e = \{ v,w \}$.

Furthermore, the fact that $\duu{f_i}$ lies on $\pi_i$ implies that
the edge of weight $C$ in $\mathcal{E}_i \cap \cycle$ lies on $\pi_i^n$ if $\duu{f_i} \notin \tfs$
and lies on $\pi_i^s$ otherwise.
Consequently, in $\duu{G}-\cycle$ there is one connected component containing all vertices
of $\duu{\tfs}$ and one connected component containing all vertices of $\overline{\duu{\tfs}}$.
By the minimality of $\cycle$, we infer that $\duu{G}-\cycle$ contains 
no other connected components apart from the aforementioned two components.
By planarity, since any minimum cut in a planar graph corresponds to a collection of cycles
in its dual, this implies that $\cycle$ is a single cycle in $G$.

Let $e_i$ be the unique edge of $\mathcal{E}_i \cap \cycle$ of weight $C$
and let $e_i'$ be the unique edge of $\mathcal{E}_i \cap \cycle_{\sign{\du{\tfs}}}$ of weight $C$.
We inductively prove that $e_i = e_i'$ and
that the subpath of $\cycle$ between $e_i$ and $e_{i+1}$ is the same as on
$\cycle_{\sign{\du{\tfs}}}$.
For the base of the induction, note that $e_0 = e_0' = e$.

Consider an index $i > 0$ and the face $\du{f_i}$. If $\du{f_i} \in \du{\tfs}$, i.e., $\du{f_i}$ belongs to the north side, then $e_i$ lies south of $f_i$, that is, lies on $\pi_i^s$.
Otherwise, if $f_i \notin \tfs$, then $e_i$ lies north of $f_i$, that is, lies on $\pi_i^n$.

Let $v_a$ and $v_{ba}$ be the vertices of $\cycle_{\sign{\tfs}}$ that lie 
on $L_{i-1}$ and $L_i$, respectively. By the inductive assumption, $v_a$ is an endpoint
of $e_{i-1}' = e_{i-1}$ that lies on $\cycle$.
Let $e_i = xv_{bc}$, where $v_{bc} \in L_i$ and let $e_i' = x'v_{ba}$.
Since $\cycle$ is a cycle in $G$ that contains exactly one edge on each path $\pi_i$,
we infer that $\cycle$ contains a path between $v_a$ and $v_{bc}$ that consists of
$e_i$ and a number of edges of $\mathcal{E}_i$ of weight $c_i$.
A direct check shows that the subpath from $v_a$ to $v_{ba}$ on $\cycle_{\sign{\tfs}}$
is the unique such path with minimum number of edges of weight $c_i$.
Since the weight $c_i$ is larger than the total weight of all edges of smaller weight,
from the minimality of $\cycle$ we infer that $v_{ba} = v_{bc}$ and $\cycle$
and $\cycle_{\sign{\tfs}}$ coincide on the path from $v_a$ to $b_{ba}$.

Consequently, $\cycle$ and $\cycle_{\sign{\tfs}}$ coincide on the path from the edge $e=vw$
to the vertex $v_{\rev{\sign{\tfs}}} \in L_{k-2}$. From the minimality of $\cycle$
we infer that also the edge $\{w,v_{\rev{\sign{\tfs}}} \}$ lies on the cycle $\cycle$ and, hence,
   $\cycle = \cycle_{\sign{\tfs}}$. This completes the proof.
\end{proof}

\begin{claim}\label{clm:rank}
$\rank{A_{G,c}} \geq 2^{k-2}$. 
\end{claim}
\begin{proof}
Recall Definition~\ref{def:impcyc} and the fact that $\cycle_{\sign{\du{\tfs}}}$ is defined for every important $\tfs \subseteq \tf$.
This means that the only edge in $\mathcal{E}_{k-1}$ that belongs to $\cycle_{\sign{\du{\tfs}}}$ is the edge adjacent to $\du{v_{\rev{\sign{\du{\tfs}}}}}$. Let us consider the part of adjacency matrix where rows correspond to the cuts corresponding to $\cycle_{\sign{\du{\tfs}}}$ for important $\tfs \subset \tf$ and where columns correspond to the edges in $\mathcal{E}_{k-1}$ of weight $C$. Let us order the cuts according to $\rev{\sign{\du{\tfs}}}$ and the edges by the index  of the adjacent vertex in $L_{k-2}$ (lexicographically). Then this part of $A_{G,c}$ is an identity matrix. Hence, $\rank{A_{G,c}} \geq 2^{k-2}$. 
\end{proof}

Lemma \ref{lem:uniquecuts} and Claim~\ref{clm:rank} provide the conditions necessary for Lemma~\ref{lem:mtl} to apply. This proves our main result stated in Theorem~\ref{thm:main}. 

\section{Doubly exponential example}\label{sec:side}

\begin{figure}[tb]
\centering
\begin{tikzpicture}

\begin{scope}[shift={(0, 5)}]
\foreach \i in {1,...,5}{
	\node[circle, draw=black, fill=black] (f\i) at (0, \i-2) {};
}
\foreach \i in {1,...,4}{
	\node[circle, draw=black, fill=white] (s\i) at (0, \i+3) {};
}
\node[circle, draw=black, fill=white, label=right:{$x$}] (x) at (9, 4.5) {};
\draw (0, 1) ellipse (1cm and 2.35cm) node [left=0.3] {$\overline{S}_0$};
\draw (0, 5.5) ellipse (1cm and 1.85cm) node [left=0.3] {$S_0$};

\node[circle, draw = black, fill=black, label=above:{$u_{S_0}$}] (v) at (5,5) {};
\foreach \i in {1,...,5}{
	\draw (v) -- (f\i) node [pos=0.6, below, sloped] {\tiny $\alpha$};
}
\foreach \i in {1,...,4}{
	\draw (v) -- (s\i) node [pos=0.6, below, sloped] {\tiny $(1+\frac{1}{r} + \frac{1}{r^2})\alpha$};
}
\foreach \i in {1, 2, 3, 4, 5, 6, 7, 8}{
\node[circle, draw=black, fill=black] (v\i) at (5, \i-1) {};
}
\node[circle, draw=black, fill=white] (v6) at (5, 5) {};
\draw (5, 5.5) ellipse (1cm and 2cm) node [left=0.3] {$Z$};

\node[circle, draw=black, fill=white, label=above:{$w_{Z}$}] (zb) at (8, 7) {};
\draw[dashed] (7,7.1) arc (-170:-120:1) node [right] {$\frac{\ell}{2}$};
\draw (x) -- (zb) node [midway, above, sloped] {\tiny $\frac{\ell}{2}-1$};
\draw (v7) -- (zb) node [midway, above, sloped] {\tiny 1};
\draw (v6) -- (zb)node [midway, above, sloped] {\tiny 1};
\draw (v5) -- (zb)node [midway, above, sloped] {\tiny 1};
\draw (v8) -- (zb)node [midway, above, sloped] {\tiny 1};

\end{scope}

\draw[very thick,dashed] (-2, 3) -- (11, 3);

\begin{scope}[shift={(0, -5)}]
\foreach \i in {1,...,5}{
	\node[circle, draw=black, fill=black] (f\i) at (0, \i-2) {};
}
\foreach \i in {1,...,4}{
	\node[circle, draw=black, fill=white] (s\i) at (0, \i+3) {};
}
\node[circle, draw=black, fill=white, label=right:{$x$}] (x) at (9, 4.5) {};
\draw (0, 1) ellipse (1cm and 2.35cm) node [left=0.3] {$\overline{S}_0$};
\draw (0, 5.5) ellipse (1cm and 1.85cm) node [left=0.3] {$S_0$};

\node[circle, draw = black, fill=black, label=above:{$u_{S_0}$}] (v) at (5,5) {};
\foreach \i in {1,...,5}{
	\draw (v) -- (f\i) node [pos=0.6, below, sloped] {\tiny $\alpha$};
}
\foreach \i in {1,...,4}{
	\draw (v) -- (s\i) node [pos=0.6, below, sloped] {\tiny $(1+\frac{1}{r} + \frac{1}{r^2})\alpha$};
}
\foreach \i in {1, 2, 3, 4, 5, 6, 7, 8}{
\node[circle, draw=black, fill=black] (v\i) at (5, \i-1) {};
}
\node[circle, draw=black, fill=white] (v6) at (5, 5) {};
\draw (5, 1.5) ellipse (1cm and 2cm) node [left=0.3] {$Z$};
\node[circle, draw=black, fill=black, label=below:{$w_{Z}$}] (zw) at (8, 2){};
\draw (x) -- (zw) node [midway, above, sloped] {\tiny $\frac{\ell}{2}-1$};
\draw (v1) -- (zw) node [midway, above, sloped] {\tiny 1};
\draw (v2) -- (zw) node [midway, above, sloped] {\tiny 1};
\draw (v3) -- (zw) node [midway, above, sloped] {\tiny 1};

\end{scope}

\end{tikzpicture}
\caption{Illustration of the construction. The two panels correspond to two cases in the proof, either $u_{S_0} \in Z$ (top panel) or $u_{S_0} \notin Z$ (bottom panel).}\label{fig:double-exp}
\end{figure}

In this section we show an example graph for which the compression technique introduced by Hagerup et al~\cite{HagerupKNR98} does indeed produce a mimicking network on
roughly $2^{\binom{k-1}{\lfloor (k-1)/2 \rfloor}}$ vertices.
Our example relies on doubly exponential edge costs. Note that an example with single exponential costs can be compressed into a mimicking network of size single exponential in $k$ using the techniques of~\cite{KratschW12}.

Before we go on, let us recall the technique of Hagerup et al~\cite{HagerupKNR98}. Let $G$ be a weighted graph and $Q$ be the set of terminals. Observe that a minimum cut separating $S \subset Q$ from $\overline{S}=Q \setminus S$, when removed from $G$, divides the vertices of $G$ into two sides: the side of $S$ and the side of $\overline{S}$. The side is defined for each vertex, as all connected components obtained by removing the minimum cut contain a terminal. Now if two vertices $u$ and $v$
are on the same side of the minimum cut between $S$ and $\overline{S}$ for every $S \subset Q$, then they can be merged without changing the size of any minimum $S$-separating cut. As a result there is at most $2^{2^k}$ vertices in the graph;
as observed by~\cite{ChambersE13,KhanR14}, this bound can be improved to roughly $2^{\binom{k-1}{\lfloor (k-1)/2 \rfloor}}$. After this brief introduction we move on to describing our example.

Our construction builds up on the example provided in~\cite{KrauthgamerR13} in the proof of Theorem 1.2.
As stated in Theorem~\ref{thm:side} of this paper, our construction works for parameter $k$ equal to $6$ modulo $8$.
Let $k = 2r+2$, that is, $r$ is equal to $2$ modulo $4$.
These remainder assumptions give the following observation via standard calculations.
\begin{lemma}\label{lem:ell-even}
The integer $\ell := \binom{2r+1}{r}$ is even.
\end{lemma}
\begin{proof}
Recall that $r$ equals $2$ modulo $4$.
Since $\binom{2r+1}{r} = \frac{(2r+1)!}{r!(r+1)!}$, while the largest power of $2$ that divides $a!$ equals $\sum_{i=1}^\infty \lfloor \frac{a}{2^i} \rfloor$, we have that
the largest power of $2$ that divides $\binom{2r+1}{r}$ equals:
\begin{align*}
& \sum_{i=1}^\infty \left\lfloor \frac{2r+1}{2^i} \right\rfloor - \sum_{i=1}^\infty \left\lfloor \frac{r}{2^i} \right\rfloor - \sum_{i=1}^\infty \left\lfloor \frac{r+1}{2^i} \right\rfloor 
 = r + \sum_{i=1}^\infty \left\lfloor \frac{r}{2^i} \right\rfloor - 2 \sum_{i=1}^\infty \left\lfloor \frac{r}{2^i} \right\rfloor  \\
&\quad = r - \sum_{i=1}^\infty \left\lfloor \frac{r}{2^i} \right\rfloor 
 = r - \frac{r}{2} - \frac{r-2}{4} - \sum_{i=1}^\infty \left\lfloor \frac{r}{4 \cdot 2^i} \right\rfloor 
 \geq  \frac{1}{2} + \frac{r}{4} - \sum_{i=1}^\infty \frac{r}{4 \cdot 2^i}
 = \frac{1}{2}.
\end{align*}
In particular, it is positive. This finishes the proof of the lemma.
\end{proof}

We start our construction with a complete bipartite graph $G_0 = (Q_0, U, E)$, where one side of the graph consists of $2r+1 = k-1$ terminals $Q_0$, and the other side of the graph consists of $\ell = \binom{2r+1}{r}$ non-terminals
$U = \{u_S~|~S \in \binom{Q_0}{r}\}$. That is, the vertices $u_S \in U$ are indexed by subsets of $Q_0$ of size $r$.
The cost of edges is defined as follows. Let $\alpha$ be a large constant that we define later on.
Every non-terminal $u_S$ is connected by edges of cost $\alpha$ to every terminal $q \in Q_0 \setminus S$ and by edges of cost $(1+\frac{1}{r} + \frac{1}{r^2})\alpha$ to every terminal $q \in S$.
To construct the whole graph $G$, we extend $G_0$ with a last terminal $x$ (i.e., the terminal set is $Q = Q_0 \cup \{x\}$)
 and build a third layer of $m = \binom{\ell}{\ell/2}$ non-terminal vertices $W = \{w_Z~|~Z \in \binom{U}{\ell/2}\}$. That is, the vertices $w_Z \in W$ are indexed by subsets of $U$ of size $\ell/2$.
There is a complete bipartite graph between $U$ and $W$ and every vertex of $W$ is adjacent to $x$.
The cost of edges is defined as follows. An edge $u_S w_Z$ is of cost $1$ if $u_S \in Z$, and of cost $0$ otherwise. Every edge of the form $xw_Z$ is of cost $\ell/2 - 1$.
This finishes the description of the construction. For the reference see the top picture in Figure~\ref{fig:double-exp}.

We say that a set $S \subseteq Q$ is \emph{important} if $x \in S$ and $|S| = r+1$. Note that there are $\ell = \binom{2r+1}{r} = \binom{k-1}{\lfloor (k-1)/2 \rfloor}$ important sets.
We observe the following.
\begin{lemma}\label{lem:dblexp}
Let $S \subset Q$ be important and let $S_0 = S \setminus \{x\} = S \cap Q_0$. For $\alpha > r^2 \ell|W|$,
    the vertex $w_Z$ is on the $S$ side of the minimum cut between $S$ and $Q \setminus S$
    if and only if $u_{S_0} \in Z$. 
\end{lemma}

\begin{proof}
First, note that if $\alpha > r^2 \ell |W|$, then the total cost of all the edges incident to vertices of $W$ is less than $\frac{1}{r^2} \alpha$.
Intuitively, this means that cost of the cut inflicted by the edges of $G_0$ is of absolutely higher importance than the ones incident with $W$.

Consider an important set $S \subseteq Q$ and let $S_0 = S \setminus \{x\} = S \cap Q_0$.

Let $u_{S'} \in U$. The \emph{balance} of the vertex $u_{S'}$, denoted henceforth $\beta(u_{S'})$, is the difference of the cost of edges
connecting $u_{S'}$ with $S_0$ and the ones connecting $u_{S'}$ and $Q_0 \setminus S$. Note that we have
$$\beta(u_{S_0}) = r \cdot \left(1+\frac{1}{r}+\frac{1}{r^2}\right) \alpha - \left(r+1\right) \cdot \alpha = \frac{1}{r} \alpha.$$
On the other hand, for $S' \neq S_0$, the balance of $u_{S'}$ can be estimated as follows:
$$\beta(u_{S'}) \leq (r-1) \cdot \left(1+\frac{1}{r}+\frac{1}{r^2}\right) \alpha + \alpha - r \cdot \alpha - \left(1+\frac{1}{r}+\frac{1}{r^2}\right) \alpha 
              = -\frac{r+2}{r^2} \alpha < -\frac{1}{r^2} \alpha.$$
Consequently, as $\frac{1}{r^2} \alpha$ is larger than the cost of all edges incident with $W$,
in a minimum cut separating $S$ from $Q \setminus S$, the vertex $u_{S_0}$ picks the $S$ side, while every vertex $u_{S'}$ for $S' \neq S_0$ picks the $Q \setminus S$ side.

Consider now a vertex $w_Z \in W$ and consider two cases: either $u_{S_0} \in Z$ or $u_{S_0} \notin Z$; see also Figure~\ref{fig:double-exp}.

\myparagraph{Case 1: $u_{S_0} \in Z$. } As argued above, all vertices of $U$ choose their side according to what is best in $G_0$, so $u_{S_0}$ is the only vertex in $U$ on the $S$ side.
To join the $S$ side, $w_Z$ has to cut $\ell/2-1$ edges $u_{S'} w_Z$ of cost $1$ each, inflicting a total cost of $\ell/2-1$;
note that it does not need to cut the edge $u_{S_0}w_Z$, which is of cost $1$ as $u_{S_0} \in Z$.
To join the $Q \setminus S$ side, $w_Z$ needs to cut $xw_Z$ of cost $\ell/2-1$
and $u_{S_0}w_Z$ of cost $1$, inflicting a total cost of $\ell/2$.
Consequently, $w_Z$ joins the $S$ side.

\myparagraph{Case 2: $u_S \notin Z$. } Again all vertices of $U$ choose their side according to what is best in $G$, so $u_{S_0}$ is the only vertex in $U$ on the $S$ side.
To join the $S$ side, $w_Z$ has to cut $\ell/2$ edges $u_{S'}w_Z$ of cost $1$ each, inflicting a total
cost of $\ell/2$.
To join the $Q \setminus S$ side, $w_Z$ has to cut one edge of positive cost, namely the edge $xw_Z$ of cost $\ell/2-1$.
Consequently, $w_Z$ joins the $Q \setminus S$ side.

This finishes the proof of the lemma.
\end{proof}
Lemma~\ref{lem:dblexp} shows that $G$ cannot be compressed using the technique presented in~\cite{HagerupKNR98}.
To see that let us fix two vertices $w_Z$ and $w_{Z'}$ in $W$,
and let $u_S \in Z \setminus Z'$.
Then, Lemma~\ref{lem:dblexp} shows that $w_Z$ and $w_{Z'}$ lie on different
sides of the minimum cut between $S$ and $Q \setminus S$.
Thus, $w_Z$ and $w_{Z'}$ cannot be merged.
Similar but simpler arguments show that no other pair of vertices in $G$ can be merged.
To finish the proof of Theorem~\ref{thm:side}, observe that
$$|W| = \binom{\ell}{\ell/2} = \Omega\left(2^{\ell}/\sqrt{\ell}\right) = \Omega\left(2^{\binom{k-1}{\lfloor (k-1)/2 \rfloor} - k/2}\right).$$

\bibliographystyle{abbrv}

\bibliography{references}

\end{document}